\documentclass[sigplan,10pt,dvipsnames,authorsversion,nonacm]{acmart}
\usepackage{layouts}
\usepackage[format=hang]{subcaption}
\usepackage{tikz}
\usepackage{balance}
\usepackage{enumitem}
\usepackage[skins]{tcolorbox}
\newtcolorbox{tightbox}[1]{
top=1mm, bottom=1mm, left=1mm, right=1mm, left skip=0cm, right skip=0cm, grow to left by=0.3cm,grow to right by=0.3cm, title={\textbf{#1}}
}

\setlist{noitemsep,topsep=0pt,parsep=0pt,partopsep=0pt,leftmargin=0.4cm}

\copyrightyear{2020}
\acmYear{2020}
\setcopyright{acmlicensed}
\acmConference[PaPoC '20]{7th Workshop on Principles and Practice of Consistency for Distributed Data}{April 27, 2020}{Heraklion, Greece}
\acmBooktitle{7th Workshop on Principles and Practice of Consistency for Distributed Data (PaPoC '20), April 27, 2020, Heraklion, Greece}
\acmPrice{15.00}
\acmDOI{10.1145/3380787.3393681}
\acmISBN{978-1-4503-7524-5/20/04}

\begin{document}

\title{
    Paxos vs Raft: Have we reached consensus on distributed consensus?
}

\author{Heidi Howard}
\affiliation{%
  \institution{University of Cambridge}
  \city{Cambridge}
  \country{UK}
}
\email{first.last@cl.cam.ac.uk}

\author{Richard Mortier}
\affiliation{%
  \institution{University of Cambridge}
  \city{Cambridge}
  \country{UK}
}
\email{first.last@cl.cam.ac.uk}

\begin{abstract}
Distributed consensus is a fundamental primitive for constructing fault-tolerant, strongly-consistent distributed systems.
Though many distributed consensus algorithms have been proposed, just two dominate production systems: Paxos, the traditional,  famously subtle, algorithm; and Raft, a more recent algorithm positioned as a more understandable alternative to Paxos.

In this paper, we consider the question of which algorithm, Paxos or Raft, is the better solution to distributed consensus?
We analyse both to determine exactly how they differ  by describing a simplified Paxos algorithm using Raft's terminology and pragmatic abstractions.

We find that both Paxos and Raft take a very similar approach to distributed consensus, differing only in their approach to leader election.
Most notably, Raft only allows servers with up-to-date logs to become leaders, whereas Paxos allows any server to be leader provided it then updates its log to ensure it is up-to-date.
Raft's approach is surprisingly efficient given its simplicity as, unlike Paxos, it does not require log entries to be exchanged during leader election.
We surmise that much of the understandability of Raft comes from the paper's clear presentation rather than being  fundamental to the underlying algorithm being presented.
\end{abstract}



\maketitle

\section{Introduction}
\label{s:intro}

State machine replication~\cite{Schneider:CS90} is widely used to compose a set of unreliable hosts into a single reliable service that can provide strong consistency guarantees including linearizability~\cite{Herlihy:PLS90}.
As a result, programmers can treat a service implemented using replicated state machines as a single system, making it easy to reason about expected behaviour.
State machine replication requires that each state machine receives the same operations in the same order, which can be achieved by distributed consensus.

The Paxos algorithm~\cite{Lamport:TCS98} is synonymous with distributed consensus.
Despite its success, Paxos is famously difficult to understand, making it hard to reason about, implement correctly, and safely optimise.
This is evident in the numerous attempts to explain the algorithm in simpler terms~\cite{Lampson:IWDA96,Prisco:IWDA97,Lampson:PODC01,Lamport:SIGACT01,Boichat:SIGACT03,Meling:ICPDS13,Renesse:CS15}, and was the motivation behind  Raft~\cite{Ongaro:ATC14}.

Raft's authors' claim that Raft is as efficient as Paxos whilst being more understandable and thus provides a better foundation for building practical systems.
Raft seeks to achieve this in three distinct ways:
\begin{description}
    \item[Presentation] Firstly, the Raft paper introduces a new abstraction for describing leader-based consensus in the context of state machine replication.
    This pragmatic presentation has proven incredibly popular with engineers.
    \item[Simplicity] Secondly, the Raft paper prioritises simplicity over performance.
    For example, Raft decides log entries in-order whereas Paxos typically allows out-of-order decisions but requires an extra protocol for filling the log gaps which can occur as a result.
    \item[Underlying algorithm] Finally, the Raft algorithm takes a novel approach to leader election which alters how a leader is elected and thus how safety is guaranteed.
\end{description}
Raft rapidly became popular~\cite{Raft} and production systems today are divided between those which use Paxos~\cite{Burrows:OSDI06,Baker:CIDR11,Schwarzkopf:Eurosys13,Verma:Eurosys15,Ramakrishnan:ICMD17,Zheng:VLDB17} and those which use Raft~\cite{Atomix,Consul,Etcd,Kubernetes,Trillian,M3,Cockroach}.

To answer the question of  which, Paxos or Raft, is the better solution to distributed consensus, we must first answer the question of how exactly the two algorithms differ in their approach to consensus?
Not only will this help in evaluating these algorithms, it may also allow Raft to benefit from the decades of research optimising Paxos' performance~\cite{Gafni:ICDC:00,Lamport:DSN04,Lamport:MStechreport05,Lamport:MStechreport06,Camargos:PODC07,Moraru:SOSP13,Kraska:Eurosys13,Moraru:CC14} and vice versa~\cite{Zhang:APNet17,Arora:HotCloud17}.

However, answering this question is not a straightforward matter.
Paxos is often regarded not as a single algorithm but as a family of algorithms for solving distributed consensus.
Paxos' generality (or underspecification, depending on your point of view) means that descriptions of the algorithm vary, sometimes considerably, from paper to paper.

To overcome this problem, we present here a simplified version of Paxos that results from surveying the various published descriptions of Paxos.
This algorithm, which we refer to simply as Paxos, corresponds more closely to how Paxos is used today than to how it was first described~\cite{Lamport:TCS98}.
It has been referred to elsewhere  as \emph{multi-decree} Paxos, or just \emph{MultiPaxos}, to distinguish it from \emph{single-decree} Paxos, which decides a single value instead of a totally-ordered sequence of values.
We also describe our simplified algorithm using the style and abstractions from the Raft paper,  allowing a  fair comparison between the two different algorithms.

We conclude that there is no significant difference in understandability between the algorithms, and that Raft's leader election is surprisingly efficient given its simplicity.

\section{Background}
\label{s:back}

This paper examines distributed consensus in the context of state machine replication.
State machine replication requires that an application's deterministic state machine is replicated across $n$ servers with each applying the same set of operations in the same order.
This is achieved using a replication log, managed by a distributed consensus algorithm, typically Paxos or Raft.

We assume that the system is \emph{non-Byzantine}~\cite{Lamport:PLS82} but we do not assume that the system is synchronous.
Messages may be arbitrarily delayed and participating servers may operate at any speed, but we assume message exchange is reliable and in-order (e.g.,~through use of TCP/IP).
We do not depend upon clock synchronisation for safety, though we must for liveness~\cite{Fischer:JACM85}.
We assume each of the $n$ servers has an unique id $s$ where $s \in \{0 .. (n-1)\}$.
We assume that operations are unique, easily achieved by adding a pair of sequence number and server id to each operation.

\section{Approach of Paxos \& Raft}

Many consensus algorithms, including Paxos and Raft, use a leader-based approach to solve distributed consensus.
At a high-level, these algorithms operate as follows:

One of the $n$ servers is designated the \emph{leader}.
All operations for the state machine are sent to the leader.
The leader appends the operation to their log and asks the other servers to do the same.
Once the leader has received acknowledgements from a majority of servers that this has taken place, it applies the operation to its state machine.
This process repeats until the leader fails.
When the leader fails, another server  takes over as leader.
This process of electing a new leader involves at least a majority of servers, ensuring that the new leader will not overwrite any previously applied operations.

We now examine Paxos and Raft in more detail.
Readers may find it helpful to refer to the summaries of Paxos and Raft provided in Appendices~\ref{appendix:paxos} \&~\ref{appendix:raft}.
We focus here on the core elements of Paxos and Raft and, due to space constraints, do not compare garbage collection, log compaction, read operations or reconfiguration algorithms.

\subsection{Basics}

\begin{figure}
    \centering
    \begin{tikzpicture}[>=stealth,shorten >=1pt,auto,node distance=3.7cm, semithick, 
state/.style={fill=gray!10,draw=black,circle,font=\small},
transition/.style={->,align=center,font=\small}] 

\node[state] (Follower) {Follower};
\node[state] (Candidate) [right of=Follower]{Candidate};
\node[state] (Leader) [right of= Candidate] {Leader};

\coordinate (start) at (-0.75,2.5);

\draw(start) edge[transition] node {starts up/\\recovers} (Follower);
\draw(Follower) edge[bend left,transition] node {times out,\\starts election} (Candidate);
\draw(Candidate) edge[loop above,transition,Cyan] node {times out,\\new election} (Candidate);
\draw(Candidate) edge[bend left,transition] node {receives votes from\\ majority of servers} (Leader);
\draw(Candidate) edge[transition,above] node {discovers new} (Follower);
\draw(Candidate) edge[transition] node {term \textcolor{Cyan}{(or leader)}} (Follower);
\draw(Leader) edge[bend left,transition] node {discovers new term} (Follower);

\end{tikzpicture}
    \caption{\label{fig:states}State transitions between the server states for Paxos \& Raft. The transitions in blue are specific to Raft.}
\end{figure}
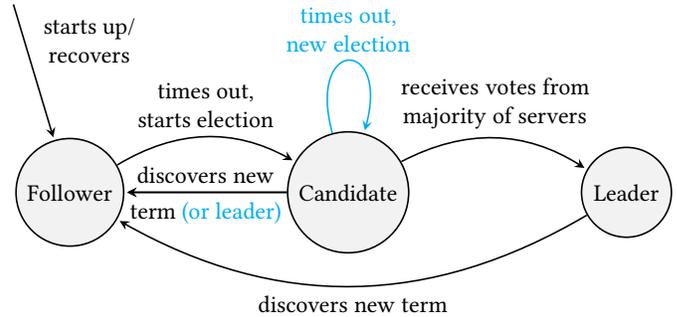

\begin{figure*}
    \centering
    \begin{subfigure}{.23\textwidth}
        \centering
        \begin{tikzpicture}[shorten >=1pt,auto, ultra thick,tx/.style={semithick,rectangle,draw=black,align=center,minimum size=0.8cm,font=\Small,inner sep=0cm}]

\def\ygap{1.1};
\def\width{0.8};

\def \entries {{1/A,2/B}, 
			   {1/A,2/B,2/C},
			   {1/A,3/B,3/D,3/E}};

\def \maxlogindex {4};

\def\commitedEntries{1/1,2/2,3/1};

\foreach \lst [count=\y from 1] in \entries {
 \node (\y) at (-0.5,({-\y+1)}*\ygap) [left] {\Small$s_\y$};
	\foreach \term/\message [count=\x from 0] in \lst {
		\node[tx] at (\x*\width,({-\y+1)}*\ygap) {\term \\\message};}}

\foreach \replica/\index in \commitedEntries{
    \draw[-] (({\index-0.5)}*\width,({-\replica+1.5)}*\ygap) -- (({\index-0.5)}*\width,({-\replica+0.5)}*\ygap);
}

\foreach \i [count=\x from 0] in {1,...,\maxlogindex} {
	\node at (\x*\width,1) {\Small\i};
}

\end{tikzpicture}
        \caption{\label{fig:log/before}Initial state of logs
          before electing a new leader.}
    \end{subfigure}
    \quad
    \begin{subfigure}{.23\textwidth}
        \centering
        \begin{tikzpicture}[shorten >=1pt,auto, ultra thick,tx/.style={semithick,rectangle,draw=black,align=center,minimum size=0.8cm,font=\Small,inner sep=0cm}]

\def\ygap{1.1};
\def\width{0.8};

\def \entries {{1/A,\textcolor{Red}{4}/B,\textcolor{Red}{4}/\textcolor{Red}{C}}, 
			   {1/A,2/B,2/C},
			   {1/A,3/B,3/D,3/E}};

\def \maxlogindex {4};

\def\commitedEntries{1/1,2/2,3/1};

\foreach \lst [count=\y from 1] in \entries {
 \node (\y) at (-0.5,({-\y+1)}*\ygap) [left] {\Small$s_\y$};
	\foreach \term/\message [count=\x from 0] in \lst {
		\node[tx] at (\x*\width,({-\y+1)}*\ygap) {\term \\\message};}}

\foreach \replica/\index in \commitedEntries{
    \draw[-] (({\index-0.5)}*\width,({-\replica+1.5)}*\ygap) -- (({\index-0.5)}*\width,({-\replica+0.5)}*\ygap);
}

\foreach \i [count=\x from 0] in {1,...,\maxlogindex} {
	\node at (\x*\width,1) {\Small\i};
}

\end{tikzpicture}
        \caption{\label{fig:log/after-s1}$s_1$ is leader in term $4$
          after a vote from $s_2$.}
    \end{subfigure}
    \quad
    \begin{subfigure}{.23\textwidth}
        \centering
        \begin{tikzpicture}[shorten >=1pt,auto, ultra thick,tx/.style={semithick,rectangle,draw=black,align=center,minimum size=0.8cm,font=\Small,inner sep=0cm}]

\def\ygap{1.1};
\def\width{0.8};

\def \entries {{1/A,2/B}, 
			   {1/A,2/B,\textcolor{Red}{5}/\textcolor{Red}{D},\textcolor{Red}{5}/\textcolor{Red}{E}},
			   {1/A,3/B,3/D,3/E}};

\def \maxlogindex {4};

\def\commitedEntries{1/1,2/2,3/1};

\foreach \lst [count=\y from 1] in \entries {
 \node (\y) at (-0.5,({-\y+1)}*\ygap) [left] {\Small$s_\y$};
	\foreach \term/\message [count=\x from 0] in \lst {
		\node[tx] at (\x*\width,({-\y+1)}*\ygap) {\term \\\message};}}

\foreach \replica/\index in \commitedEntries{
    \draw[-] (({\index-0.5)}*\width,({-\replica+1.5)}*\ygap) -- (({\index-0.5)}*\width,({-\replica+0.5)}*\ygap);
}

\foreach \i [count=\x from 0] in {1,...,\maxlogindex} {
	\node at (\x*\width,1) {\Small\i};
}

\end{tikzpicture}
        \caption{\label{fig:log/after-s2}$s_2$ is leader in term $5$
          after a vote from $s_3$.}
    \end{subfigure}
    \quad
    \begin{subfigure}{.23\textwidth}
        \centering
        \begin{tikzpicture}[shorten >=1pt,auto, ultra thick,tx/.style={semithick,rectangle,draw=black,align=center,minimum size=0.8cm,font=\Small,inner sep=0cm}]

\def\ygap{1.1};
\def\width{0.8};

\def \entries {{1/A,2/B}, 
			   {1/A,2/B,2/C},
			   {1/A,\textcolor{Red}{6}/B,\textcolor{Red}{6}/D,\textcolor{Red}{6}/E}};

\def \maxlogindex {4};

\def\commitedEntries{1/1,2/2,3/1};

\foreach \lst [count=\y from 1] in \entries {
 \node (\y) at (-0.5,({-\y+1)}*\ygap) [left] {\Small$s_\y$};
	\foreach \term/\message [count=\x from 0] in \lst {
		\node[tx] at (\x*\width,({-\y+1)}*\ygap) {\term \\\message};}}

\foreach \replica/\index in \commitedEntries{
    \draw[-] (({\index-0.5)}*\width,({-\replica+1.5)}*\ygap) -- (({\index-0.5)}*\width,({-\replica+0.5)}*\ygap);
}

\foreach \i [count=\x from 0] in {1,...,\maxlogindex} {
	\node at (\x*\width,1) {\Small\i};
}

\end{tikzpicture}
        \caption{\label{fig:log/after-s3}$s_3$ is leader in term $6$
          after a vote from $s_1$ or $s_2$.}
    \end{subfigure}
    \caption{Logs of three servers running Paxos. Figure~(\protect\subref{fig:log/before}) shows the logs when a leader election was triggered. Figures~(\protect\subref{fig:log/after-s1}---\protect\subref{fig:log/after-s3}) show the logs after a leader has been elected but before it has sent its first AppendEntries RPC. The black line shows the commit index and red text highlights the log changes.}
    \label{fig:logs}
\end{figure*}
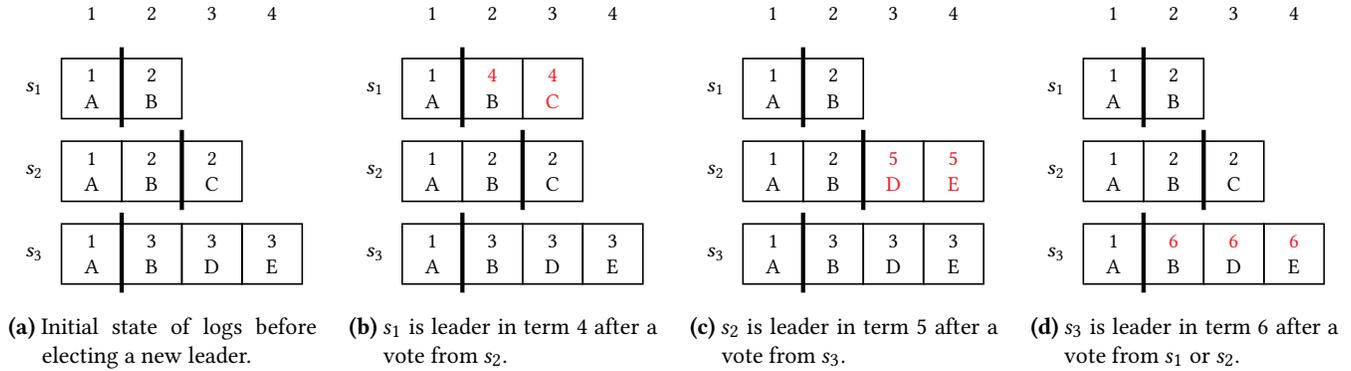

As shown in Figure~\ref{fig:states}, at any time a server can be in one of three states:
\begin{description}
    \item[Follower] A passive state where it is responsible only for replying to RPCs.
    \item[Candidate] An active state where it is trying to become a leader using the \emph{RequestVotes} RPC.
    \item[Leader] An active state where it is responsible for adding operations to the replicated log using the AppendEntries RPC.
\end{description}

Initially, servers are in the \emph{follower} state.
Each server continues as a follower until it believes that the leader has failed.
The follower then becomes a \emph{candidate} and tries to be elected leader using RequestVote RPCs.
If successful, the candidate becomes a leader.
The new leader must regularly send AppendEntries RPCs as keepalives to prevent followers from timing out and becoming candidates.

Each server stores a natural number,  the \emph{term}, which increases monotonically over time.
Initially, each server has a current term of zero.
The sending server's (hereafter, the sender) current term is included in each RPC.
When a server receives an RPC, it (hereafter, the server) first checks the included term.
If the sender's term is greater than the server's, then the server will update its term before responding to the RPC and,
if the server was either a candidate or a leader,  step down to become a follower.
If the sender's term is equal to that of the  server, then the server will respond to the RPC as usual.
If the sender's term is less than that of the server, then the server will respond negatively to the sender, including its term in the response.
When the sender receives such a response, it will step down to follower and update its term.

\subsection{Normal operation}

When a leader receives an operation, it appends it to the end of its log with the current term.
The pair of operation and term are known as a \emph{log entry}.
The leader then sends AppendEntries RPCs to all other servers with the new log entry.
Each server maintains a commit index to record which log entries are safe to apply to its state machine, and responds to the leader acknowledging successful receipt of the new log entry.
Once the leader receives positive responses from a majority of servers, the leader updates its commit index and applies the operation to its state machine.
The leader then includes the updated commit index in subsequent AppendEntries RPCs.

A follower will only append a log entry (or set of log entries) if its log prior to that entry (or entries) is identical to the leader's log.
This ensures that log entries are added in-order, preventing gaps in the log,  and ensuring followers  apply the correct log entries to their state machines.

\subsection{Handling leader failures}

This process continues until the leader fails, requiring a new leader to be established.
Paxos and Raft take different approaches to this process so we describe each in turn.

{\bf Paxos}.
A follower will timeout after failing to receive a recent AppendEntries RPC from the leader.
It then becomes a candidate and updates its term to the next term such that $t\;\mathrm{mod}\; n = s$ where $t$ is the next term, $n$ is the number of servers and $s$ is the candidate's server id.
The candidate will send RequestVote RPCs to the other servers.
This RPC includes the candidate's new term and commit index.
When a server receives the RequestVote RPC, it will respond positively provided the candidate's term is greater than its own.
This response also includes any log entries that the server has in its log subsequent to the candidate's commit index.

Once the candidate has received positive RequestVote responses from a majority of servers, the candidate must ensure its log includes all committed entries before becoming a leader.
It does so as follows.
For each index after the commit index, the leader reviews the log entries it has received alongside its own log.
If the candidate has seen a log entry for the index then it will update its own log with the entry and the new term.
If the leader has seen multiple log entries for the same index then it will update its own log with the entry from the greatest term and the new term.
An example of this is given in Figure~\ref{fig:logs}.
The candidate can now become a leader and begin replicating its log to the other servers.

{\bf Raft}.
At least one of the followers will timeout after not receiving a recent AppendEntries RPC for the leader.
It will become a candidate and increment its term.
The candidate will send RequestVote RPCs to the other servers.
Each includes the candidate's term as well as the candidate's last log term and index.
When a server receives the RequestVote request it will respond positively provided the candidate's term is greater than or equal to its own, it has not yet voted for a candidate in this term, and the candidate's log is at least as up-to-date as its own.
This last criterion can be checked by ensuring that the candidate's last log term is greater than the server's or, if they are the same, that the candidate's last index is greater than the server's.

Once the candidate has received positive RequestVote responses from a majority of servers, the candidate can become a leader and start replicating its log.
However, for safety Raft requires that the leader does not update its commit index until at least one log entry from the new term has been committed.

As there may be multiple candidates in a given term, votes may be split such that no candidate has a majority.
In this case, the candidate times out and starts a new election with the next term.

\subsection{Safety}

Both algorithms guarantee the following property:
\begin{theorem}[State Machine Safety]
If a server has applied a log entry at a given index to its state machine, no other server will ever apply a different log entry for the same index.
\end{theorem}

There is at most one leader per term and a leader will not overwrite its own log so we can prove this by proving the following:

\begin{theorem}[Leader Completeness]
If an operation $op$ is committed at index $i$ by a leader in term $t$ then all leaders of terms $>t$ will also have operation $op$ at index $i$.
\end{theorem}

\begin{proof}[Proof sketch for Paxos]

Assume operation $op$ is committed at index $i$ with term $t$.
We will use proof by induction over the terms greater than $t$.

\emph{\textsc{Base case:} If there is a leader of term $t+1$, then it will have the operation $op$ at index $i$.}

As any two majority quorums intersect and messages are ordered by term, then at least one server with operation $op$ at index $i$ and with term $t$ must have positively replied to the RequestVote RPC from the leader of $t+1$.
This server cannot have deleted or overwritten this operation as it cannot have positively responded to an AppendEntries RPC from any other leader since the leader of term $t$.
The leader will choose the operation $op$ as it will not receive any log entries with a term $>t$.

\emph{\textsc{Inductive case:} Assume that any leaders of terms $t+1$ to $t+k$ have the operation $op$ at index $i$. If there is a leader of term $t+k+1$ then it will also have operation $op$ at index $i$.}

As any two majority quorums intersect and messages are ordered by term, then at least one server with operation $op$ at index $i$ and with term $t$ to $t+k$ must have positively replied to the RequestVote RPC from the leader of term $t+k+1$.
This is because the server cannot have deleted or overwritten this operation as it has not positively responded to an AppendEntries RPC from any leader except those with terms $t$ to $t+k$.
From our induction hypothesis, all these leaders will also have operation $op$ at index $i$ and thus will not have overwritten it.
The leader may choose another operation only if it receives a log entry with that different operation at index $i$ and with a greater term.
From our induction hypothesis, all leaders of terms $t$ to $t+k$ will also have operation $op$ at index $i$ and thus will not write another operation at index $i$.
\end{proof}

The proof for Raft uses the same induction but the details differ due to Raft's different approach to leader election.

\section{Discussion}

\begin{table*}[t]
\centering
\begin{tabular}{p{3.5cm} p{6.5cm} p{6.5cm} }
 \hline
  & \textbf{Paxos} & \textbf{Raft} \\
 \hline
 \textbf{How does it ensure that each term has at most one leader?} &
 A server $s$ can only be a candidate in a term $t$ if $t\;\mathrm{mod}\; n = s$.
 There will only be one candidate per term so only one leader per term. &
 A follower can become a candidate in any term.
 Each follower will only vote for one candidate per term, so only one candidate can get a majority of votes and become the leader.\\
 \hline
 \textbf{How does it ensure that a new leader's log contains all committed log entries?} &
 Each RequestVote reply includes the follower's log entries.
 Once a candidate has received RequestVote responses from a majority of followers, it adds the entries with the highest term to its log.&
 A vote is granted only if the candidate's log is at least as up-to-date as the followers'. This ensures that a candidate only becomes a leader if its log is at least as up-to-date as a majority of followers. \\
 \hline
 \textbf{How does it ensure that leaders safely commit log entries from previous terms?} &
 Log entries from previous terms are added to the leader's log with the leader's term.
 The leader then replicates the log entries as if they were from the leader's term. &
 The leader replicates the log entries to the other servers without changing the term. The leader cannot consider these entries committed until it has replicated a subsequent log entry from its own term.
 \\
 \hline
\end{tabular}
\caption{\label{table:differences}Summary of the differences between Paxos and Raft}
\end{table*}

Raft and Paxos take different approaches to leader election,  summarised in Table~\ref{table:differences}.
We compare two dimensions, understandability and efficiency, to determine which is best.

\paragraph{Understandability.}
Raft guarantees that if two logs contain the same operation then it will have the same index and term in both.
In other words, each operation is assigned a unique index and term pair.
However, this is not the case in Paxos, where an operation may be assigned a higher term by a future leader, as demonstrated by operations B and C in Figure~\ref{fig:log/after-s1}.
In Paxos, a log entry before the commit index may be overwritten.
This is safe because the log entry will only be overwritten by an entry with the same operation, but it not as intuitive as Raft's approach.

The flip side of this is that Paxos makes it safe to commit a log entry if it is present on a majority of servers; but
this is not the case for Raft, which requires that a leader only commits a log entry from a previous term if it is present on the majority of servers and the leader has committed a subsequent log entry from the current term.

In Paxos, the log entries replicated by the leader are either from the current term or they are already committed.
We can see this in Figure~\ref{fig:logs}, where all log entries after the commit index on the leader have the current term.
This is not the case in Raft where a leader may be replicating uncommitted entries from previous terms.

Overall, we feel that Raft's approach is slightly more understandable than Paxos' but not significantly so.

\paragraph{Efficiency.}
In Paxos, if multiple servers become candidates simultaneously, the candidate with the higher term will win the election.
In Raft, if multiple servers become candidates simultaneously, they may split the votes as they will have the same term, and so neither will win the election.
Raft mitigates this by having followers wait an additional period, drawn  from a uniform random distribution, after the election timeout.
We thus expect that Raft will be both slower and have higher variance in the time taken to elect a leader.

However, Raft's leader election phase is more lightweight than Paxos'.
Raft only allows a candidate with an up-to-date log to become a leader and thus need not send log entries during leader election.
This is not true of Paxos, where every positive RequestVote response includes the follower's log entries after the candidate's commit index.
There are various options to reduce the number of log entries sent but ultimately, it will always be necessary for some log entries to be sent if the leader's log is not already up to date.

It is not just with the RequestVote responses that Paxos sends more log entries than Raft.
In both algorithms, once a candidate becomes a leader it will copy its log to all other servers.
In Paxos, a log entry may have been given a new term by the leader and thus the leader may send another copy of the log entry to a server which already has a copy.
This is not the case with Raft, where each log entry keeps the same term throughout its time in the log.

Overall, compared to Paxos, Raft's approach to leader election is surprisingly efficient for such a simple approach.

\section{Relation to classical Paxos}
\label{s:paxos_comparision}

Readers who are familiar with Paxos may feel that our description of Paxos differs from those previously published, and so we now outline how our Paxos algorithm relates to those found elsewhere in the literature.

\paragraph{Roles}
Some descriptions of Paxos divide the responsibility of Paxos into three roles: proposer, acceptor and learner~\cite{Lamport:SIGACT01} or leader, acceptor and replica~\cite{Renesse:CS15}.
Our presentation uses just one role, server, which incorporates all roles. This  presentation using a single role has  also used the name \emph{replica}~\cite{Chandra:SOSP07},

\paragraph{Terminology}
Terms are also referred to as views, ballot numbers~\cite{Renesse:CS15}, proposal numbers~\cite{Lamport:SIGACT01}, round numbers or sequence numbers~\cite{Chandra:SOSP07}.
Our leader is also referred to as a master~\cite{Chandra:SOSP07}, primary, coordinator or distinguished proposer~\cite{Lamport:SIGACT01}.
Typically, the period during which a server is a candidate is known as phase-1 and the period during which a server is a leader is known as phase-2.
The RequestVote RPCs are often referred to as phase1a and phase1b messages~\cite{Renesse:CS15}, prepare request and response~\cite{Lamport:SIGACT01} or prepare and promise messages.
The AppendEntries RPCs are often referred to as phase2a and phase2b messages~\cite{Renesse:CS15}, accept request and response~\cite{Lamport:SIGACT01} or propose and accept messages.

\paragraph{Terms}
Paxos requires only that terms are totally ordered and that each server is allocated a disjoint set of terms (for safety) and that each server can use a term greater than any other term (for liveness).
Whilst some descriptions of Paxos use round-robin natural numbers like  us~\cite{Chandra:SOSP07}, others use lexicographically ordered pairs, consisting of an integer and the server ID, where each server only uses terms containing its own ID~\cite{Renesse:CS15}.


\paragraph{Ordering}
Our log entries are replicated and decided in-order. This is not necessary but it does avoid the complexities of filling log gaps~\cite{Lamport:SIGACT01}.
Similarly, some descriptions of Paxos bound the number of concurrent decisions,  often necessary for reconfiguration~\cite{Lamport:SIGACT01,Renesse:CS15}.

\section{Summary}
\label{s:summary}

The Raft algorithm was proposed to address the longstanding issues with understandability of the widely studied Paxos algorithm. In this paper, we have demonstrated that much of the understandability of Raft comes from its pragmatic abstraction and excellent presentation.
By describing a simplified Paxos algorithm using the same approach as Raft, we find that the two algorithms differ only in their approach to leader election.
Specifically:
\begin{enumerate}[label=(\roman*)]
    \item Paxos divides terms between servers, whereas Raft allows a follower to become a candidate in any term but followers will vote for only one candidate per term.
    \item Paxos followers will vote for any candidate, whereas Raft followers will only vote for a candidate if the candidate's log is at-least-as up-to-date.
    \item If a leader has uncommitted log entries from a previous term, Paxos will replicate them in the current term whereas Raft will replicate them in their original term.
\end{enumerate}

The Raft paper claims that Raft is significantly more understandable than Paxos, and  as efficient.
On the contrary, we find that the two algorithms are not significantly different in understandability but Raft's leader election is surprisingly lightweight when compared to Paxos'.
Both algorithms we have presented are  na\"{\i}ve by design and could certainly be optimised to improve performance, though often at the cost of increased complexity.

\paragraph{Acknowledgements.}
This work funded in part by EPSRC EP/N028260/2 and EP/M02315X/1.

{
  \bibliographystyle{acm}
  \bibliography{refs}
  \balance{}
  \clearpage
}

\appendix

\section{Paxos Algorithm}
\label{appendix:paxos}

This summarises our simplified, Raft-style Paxos algorithm.
The text in red is unique to Paxos.
\begin{tightbox}{State}
\textbf{\underline{Persistent state on all servers:}}
(Updated on stable storage before responding to RPCs)
\begin{description}
    \item[currentTerm] latest term server has seen (initialized to 0 on first boot, increases monotonically)
    \item[log[]] log entries; each entry contains command for state machine, and term when entry was received by leader (first index is 1)
\end{description}

\textbf{\underline{Volatile state on all servers:}}
\begin{description}
    \item[commitIndex] index of highest log entry known to be committed (initialized to 0, increases monotonically)
    \item[lastApplied] index of highest log entry applied to state machine (initialized to 0, increases monotonically)
\end{description}

\textcolor{Red}{\textbf{\underline{Volatile state on candidates:}}}
\textcolor{Red}{(Reinitialized after election)}
\begin{description}
    \item[\textcolor{Red}{entries[]}] \textcolor{Red}{Log entries received with votes}
\end{description}

\textbf{\underline{Volatile state on leaders:}}
(Reinitialized after election)
\begin{description}
    \item[nextIndex[]] for each server, index of the next log entry to send to that server \textcolor{Red}{(initialized to leader commit index + 1)}
    \item[matchIndex[]] for each server, index of highest log entry known to be replicated on server (initialized to 0, increases monotonically)
\end{description}
\end{tightbox}

\begin{tightbox}{AppendEntries RPC}
Invoked by leader to replicate log entries; also used as heartbeat

\textbf{\underline{Arguments:}}
\begin{description}
    \item[term] leader's term
    \item[prevLogIndex] index of log entry immediately preceding new ones
    \item[prevLogTerm] term of prevLogIndex entry
    \item[entries[]] log entries to store (empty for heartbeat; may send more than one for efficiency)
    \item[leaderCommit] leader's commitIndex 
\end{description}
\textbf{\underline{Results:}}
\begin{description}
    \item[term] currentTerm, for leader to update itself
    \item[success] true if follower contained entry matching prevLogIndex and prevLogTerm
\end{description}
\textbf{\underline{Receiver implementation:}}
\begin{enumerate}
    \item Reply false if term < currentTerm
    \item Reply false if log doesn't contain an entry at prevLogIndex
whose term matches prevLogTerm
    \item If an existing entry conflicts with a new one (same index
but different terms), delete the existing entry and all that
follow it 
    \item Append any new entries not already in the log
    \item If leaderCommit > commitIndex: set commitIndex = \\min(leaderCommit, index of last new entry)
\end{enumerate}
\end{tightbox}

\begin{tightbox}{RequestVote RPC}
Invoked by candidates to gather votes

\textbf{\underline{Arguments:}}
\begin{description}
    \item[term] candidate's term
    \item[\textcolor{Red}{leaderCommit}] \textcolor{Red}{candidate's commit index}
\end{description}
\textbf{\underline{Results:}}
\begin{description}
    \item[term] currentTerm, for candidate to update itself
    \item[voteGranted] true indicates candidate received vote
    \item[\textcolor{Red}{entries[]}] \textcolor{Red}{follower's log entries after leaderCommit}
\end{description}
\textbf{\underline{Receiver implementation:}}
\begin{enumerate}
    \item Reply false if term < currentTerm
    \item \textcolor{Red}{Grant vote and send any log entries after leaderCommit}
\end{enumerate}
\end{tightbox}

\begin{tightbox}{Rules for Servers}
\textbf{\underline{All Servers:}}
\begin{itemize}
    \item If commitIndex > lastApplied: increment lastApplied and apply log[lastApplied] to state machine
    \item If RPC request or response contains term T > currentTerm: set currentTerm = T and convert to follower
\end{itemize}
\textbf{\underline{Followers:}}
\begin{itemize}
    \item Respond to RPCs from candidates and leaders
    \item If election timeout elapses without receiving AppendEntries
RPC from current leader or granting vote to candidate: convert to candidate
\end{itemize}
\textbf{\underline{Candidates:}}
\begin{itemize}
    \item  On conversion to candidate, start election: \textcolor{Red}{increase currentTerm to next $t$ such that $t\;\mathrm{mod}\; n = s$, copy any log entries after commitIndex to entries[]}, and send RequestVote RPCs to all other servers
    \item \textcolor{Red}{Add any log entries received from RequestVote responses to entries[]}
    \item If votes received from majority of servers: \textcolor{Red}{update log by adding entries[] with currentTerm (using value with greatest term if there are multiple entries with same index) and} become leader
\end{itemize}
\textbf{\underline{Leaders:}}
\begin{itemize}
    \item Upon election: send initial empty AppendEntries RPCs (heartbeat) to each server; repeat during idle periods to prevent election timeouts
    \item If command received from client: append entry to local log, respond after entry applied to state machine
    \item If last log index $\geq$ nextIndex for a follower: send AppendEntries RPC with log entries starting at nextIndex
    \begin{itemize}
        \item If successful: update nextIndex and matchIndex for follower
        \item If AppendEntries fails because of log inconsistency: decrement nextIndex and retry
    \end{itemize}
    \item If there exists an N such that N > commitIndex and a majority of matchIndex[i] $\geq$ N: set commitIndex = N
\end{itemize}
\end{tightbox}

\balance{}
\newpage

\section{Raft Algorithm}
\label{appendix:raft}

This is a reproduction of Figure~2 from the Raft paper~\cite{Ongaro:ATC14}.
The text in red is unique to Raft.
\begin{tightbox}{State}
\textbf{\underline{Persistent state on all servers:}}
(Updated on stable storage before responding to RPCs)
\begin{description}
    \item[currentTerm] latest term server has seen (initialized to 0 on first boot, increases monotonically)
    \item[\textcolor{Red}{votedFor}] \textcolor{Red}{candidateId that received vote in current term (or null if none)}
    \item[log[]] log entries; each entry contains command for state machine, and term when entry was received by leader (first index is 1)
\end{description}

\textbf{\underline{Volatile state on all servers:}}
\begin{description}
    \item[commitIndex] index of highest log entry known to be committed (initialized to 0, increases monotonically)
    \item[lastApplied] index of highest log entry applied to state machine (initialized to 0, increases monotonically)
\end{description}

\textbf{\underline{Volatile state on leaders:}}
(Reinitialized after election)
\begin{description}
    \item[nextIndex[]] for each server, index of the next log entry to send to that server \textcolor{Red}{(initialized to leader last log index + 1)}
    \item[matchIndex[]] for each server, index of highest log entry known to be replicated on server (initialized to 0, increases monotonically)
\end{description}
\end{tightbox}
\begin{tightbox}{AppendEntries RPC}
Invoked by leader to replicate log entries; also used as heartbeat

\textbf{\underline{Arguments:}}
\begin{description}
    \item[term] leader's term
    \item[prevLogIndex] index of log entry immediately preceding new ones
    \item[prevLogTerm] term of prevLogIndex entry
    \item[entries[]] log entries to store (empty for heartbeat; may send more than one for efficiency)
    \item[leaderCommit] leader's commitIndex 
\end{description}
\textbf{\underline{Results:}}
\begin{description}
    \item[term] currentTerm, for leader to update itself
    \item[success] true if follower contained entry matching prevLogIndex and prevLogTerm
\end{description}
\textbf{\underline{Receiver implementation:}}
\begin{enumerate}
    \item Reply false if term < currentTerm
    \item Reply false if log doesn't contain an entry at prevLogIndex
whose term matches prevLogTerm
    \item If an existing entry conflicts with a new one (same index
but different terms), delete the existing entry and all that
follow it 
    \item Append any new entries not already in the log
    \item If leaderCommit > commitIndex: set commitIndex = \\min(leaderCommit, index of last new entry)
\end{enumerate}
\end{tightbox}

\begin{tightbox}{RequestVote RPC}
Invoked by candidates to gather votes

\textbf{\underline{Arguments:}}
\begin{description}
    \item[term] candidate's term
    \item[\textcolor{Red}{candidateId}] \textcolor{Red}{candidate requesting vote}
    \item[\textcolor{Red}{lastLogIndex}] \textcolor{Red}{index of candidate's last log entry}
    \item[\textcolor{Red}{lastLogTerm}] \textcolor{Red}{term of candidate's last log entry}
\end{description}
\textbf{\underline{Results:}}
\begin{description}
    \item[term] currentTerm, for candidate to update itself
    \item[voteGranted] true indicates candidate received vote
\end{description}
\textbf{\underline{Receiver implementation:}}
\begin{enumerate}
    \item Reply false if term < currentTerm
    \item \textcolor{Red}{If votedFor is null or candidateId, and candidate's log is at least as up-to-date as receiver's log: grant vote}
\end{enumerate}
\end{tightbox}

\begin{tightbox}{Rules for Servers}
\textbf{\underline{All Servers:}}
\begin{itemize}
    \item If commitIndex > lastApplied: increment lastApplied, apply log[lastApplied] to state machine
    \item If RPC request or response contains term T > currentTerm: set currentTerm = T and convert to follower
\end{itemize}
\textbf{\underline{Followers:}}
\begin{itemize}
    \item Respond to RPCs from candidates and leaders
    \item If election timeout elapses without receiving AppendEntries
RPC from current leader or granting vote to candidate: convert to candidate
\end{itemize}
\textbf{\underline{Candidates:}}
\begin{itemize}
    \item  On conversion to candidate, start election: \textcolor{Red}{increment currentTerm, vote for self, reset election timer} and send RequestVote RPCs to all other servers
    \item If votes received from majority of servers: become leader
    \item \textcolor{Red}{If AppendEntries RPC received from new leader: convert to follower}
    \item \textcolor{Red}{If election timeout elapses: start new election}
\end{itemize}
\textbf{\underline{Leaders:}}
\begin{itemize}
    \item Upon election: send initial empty AppendEntries RPCs (heartbeat) to each server; repeat during idle periods to prevent election timeouts
    \item If command received from client: append entry to local log, respond after entry applied to state machine
    \item If last log index $\geq$ nextIndex for a follower: send AppendEntries RPC with log entries starting at nextIndex
    \begin{itemize}
        \item If successful: update nextIndex and matchIndex for follower
        \item If AppendEntries fails because of log inconsistency: decrement nextIndex and retry
    \end{itemize}
    \item If there exists an N such that N > commitIndex and a majority of matchIndex[i] $\geq$ N, \textcolor{Red}{and log[N].term $==$ currentTerm}: set commitIndex = N
\end{itemize}
\end{tightbox}

\balance{}
\end{document}